\documentclass[a4paper,twoside]{article}

\usepackage[T1]{fontenc}
\usepackage{epsfig}
\usepackage{subcaption}
\usepackage{calc}
\usepackage{amssymb}
\usepackage{amstext}
\usepackage{mathrsfs, amsmath}
\usepackage{amsthm}
\usepackage{multicol}
\usepackage{pslatex}
\usepackage{apalike}
\usepackage{hyperref}
\usepackage{SCITEPRESS}     
\usepackage{boldline}
\usepackage[ruled,vlined]{algorithm2e}

\DeclareMathOperator*{\argmin}{argmin}

\newtheorem{proposition}{Proposition}

\begin{document}

\title{Image Restoration using Plug-and-Play CNN MAP Denoisers}

\author{\authorname{Siavash Bigdeli\sup{1}, David Honz\'atko\sup{1},
Sabine S\"usstrunk\sup{2} and L. Andrea Dunbar\sup{1}}
\affiliation{\sup{1}Centre Suisse d'Electronique et de Microtechnique (CSEM), Neuch\^atel, Switzerland}
\affiliation{\sup{2}School of Computer and Communication Sciences, \'Ecole Polytechnique F\'ed\'erale de Lausanne (EPFL), Switzerland}
\email{\{siavash.bigdeli, david.honzatko, andrea.dunbar\}@csem.ch, sabine.susstrunk@epfl.ch}
}


\abstract{
	Plug-and-play denoisers can be used to perform generic image restoration tasks independent of the degradation type.
	These methods build on the fact that the Maximum a Posteriori (MAP) optimization can be solved using smaller sub-problems, including a MAP denoising optimization.
	We present the first end-to-end approach to MAP estimation for image denoising using deep neural networks.
	We show that our method is guaranteed to minimize the MAP denoising objective, which is then used in an optimization algorithm for generic image restoration.
	We provide theoretical analysis of our approach and show the quantitative performance of our method in several experiments.
	Our experimental results show that the proposed method can achieve 70x faster performance compared to the state-of-the-art, while maintaining the theoretical perspective of MAP.
}

\onecolumn \maketitle \normalsize \setcounter{footnote}{0} \vfill

\section{\uppercase{Introduction}}
\label{sec:introduction}
\noindent Image restoration is a classical signal processing problem with application in diverse domains such as biology, physics, and entertainment.
Due to the inherent ambiguity of this task, either the Maximum a Posteriori (MAP) or the Minimum Mean Squared Error (MMSE) estimators are usually used to produce consistent results.
With recent advances in deep learning, however, most of the methods employ the MMSE solution due to its simple loss and straightforward training.
The MMSE estimator is obtained by minimizing the euclidean distance between the results and the corresponding ground truth.
In image denoising and super-resolution, Zhang et al.~\cite{zhang2017beyond} showed that training deep neural networks using the MMSE objective can achieve state-of-the-art results.
Often the absolute norm is used in place of euclidean norm, which achieves visually more pleasant reconstructions.
However, these techniques are usually contaminated with undesired visual artifacts such as over-smoothness~\cite{isola2017image}, even though they can achieve better results in terms of peak signal to noise ratio (PSNR).

Another reason to use MAP estimators is the ability to use a single prior model to perform several image restoration tasks as a generic frame work.
This is usually done either by learning an explicit image prior, or by using plug-and-play denoisers that inherently learn the prior.
We present a novel method to obtain MAP results efficiently using deep neural networks.
In contrast to MMSE, the MAP estimator does not enforce correctness of intensity, but it optimizes for the most probable solution -- conditioned to the degraded observation.
In many applications such as medicine, where the detection overrules the correctness of signal, MAP is a better estimator than MMSE objective.

Performing image restoration model with MAP objective requires using an explicit image prior, which is usually very inefficient.
In this paper, we propose an optimization technique with improved efficiency to solve the MAP objective with an explicit image prior model. \footnote{The code and the trained models are available at: \url{https://github.com/DawyD/cnn-map-denoiser}}
In summary, the contributions of our work is as follows:
\begin{itemize}
\item A novel training strategy to learn a MAP denoiser using an end-to-end model and its neural network parametrization,
\item A generic image restoration algorithm based on Alternating direction method of multipliers (ADMM	) that uses our neural network and is more efficient in optimizing MAP inference compared to other methods using explicit priors.
\end{itemize}
The rest of the paper is organized as follows: In Section~\ref{sec:related} we discuss relevant work and the challenges of obtaining the MAP solution using neural networks.
Sections~\ref{sec:restoration},~\ref{sec:MMSEdenoising} discuss the background for image restoration using MAP and MMSE estimators.
In Section~\ref{sec:denoising}	we describe our new loss function for training neural networks to perform MAP image denoising.
Finally, we discuss and demonstrate our experimental results in Section~\ref{sec:experiments} and conclude our findings in Section~\ref{sec:conclusion}.

\section{\uppercase{Related work}}
\label{sec:related}

\noindent Several methods have been proposed recently  that use deep neural network for generic MAP image reconstruction.
Among these models, some are based on hand designed features~\cite{ulyanov2018deep}, and some are based on the explicit density of images~\cite{ulyanov2018deep,bigdeli2017deep,bigdeli2017image}.
These techniques mostly use iterative optimization based on gradient descent that take several steps to converge.
On the other hand, several plug-and-play methods have been developed for generic image restoration, but only a few methods use the explicit MAP objective~\cite{ahmad2019plug}.
Most of the proposed techniques try to benefit from the fact that a sub-problem of the image reconstruction optimization is a denoising problem and use ad-hoc denoisers to solve this sub-problem.
Some of these methods use more classical denoisers such as BM3D or Non-local means~\cite{heide2014flexisp,venkatakrishnan2013plug}.
More recent approaches use convolutional neural networks, trained as an MMSE model for Gaussian noise removal~\cite{zhang2017learning}.
Reehorst and Schniter~\cite{reehorst2018regularization} showed, however, that the denoising sub-problem is of the form of a MAP objective and will not express the original prior if replaced with other types of denoisers like MMSE.
This work focuses on developing a network model that can substitute the MAP denoising sub-problem, and use it to perform generic reconstruction of arbitrary degradation models.

Similarly, others have investigated how to provide explicit priors to replace the denoising sub-problem.
Chang et al.~\cite{chang2017one} employ a classification network by assuming a discrete manifold for natural images, where all images have the same likelihoods.
Sonderby et al.~\cite{sonderby2016amortised} makes no assumption about the image prior, but their optimization is limited to membership-based data dependencies (such as super-resolution) and cannot be generalized for other tasks after the training.

We summarize the main characteristics of related work in ~\ref{tab:comparison}.
IRCNN~\cite{zhang2017learning} and DIP~\cite{ulyanov2018deep} do not have an explicit prior, therefore their solution cannot be expressed in the framework of the MAP image restoration problem.
On the other hand DMSP~\cite{bigdeli2017deep} uses the slow gradient descent optimization that requires many iterations to converge.
We propose an algorithm that can perform up to 70x faster than DMSP and DIP by using the ADMM optimization.
In contrast to IRCNN, our prior is explicit and we require only one network throughout the optimization process.

\begin{table}[t]
	\vspace{0.2cm}
	\caption{Comparison of generic image restoration algorithms and their characteristics. Our method preserves the theoretical guarantees of the conventional MAP estimator by providing an explicit image prior, while also perform an efficient optimization.}\label{tab:comparison} \centering
	\begin{tabular}{l c c c c}
		\hlineB{3}
		Method & Prior & Optim. & Speed & \# Nets \\
		\hline
		DIP & Implicit & GD & Slow & 1 \\
		IRCNN & Implicit & HQS & Fast & 25 \\
		DMSP & Explicit & GD & Slow & 1 \\
		\textbf{Ours} & \textbf{Explicit} & \textbf{ADMM} & \textbf{Fast} & \textbf{1} \\
		\hlineB{3}
	\end{tabular}
\end{table}

\section{\uppercase{Image Restoration via MAP Objective}}
\label{sec:restoration}
\noindent We use the standard model of degradation including the noise as
\begin{equation}
y = Kx + \eta,
\end{equation}
where $x$ is the unknown sharp image, $\eta \sim \mathcal{N}(0,\sigma^2)$ is a noise vector with standard deviation $\sigma$, and $K$ is a Toeplitz matrix representing the blur kernel.
The MAP estimator intuitively finds the most probable solution of the degradation by maximizing its posterior probability.
Using our formulation of the degradation model, this leads to the following optimization
\begin{equation}
\argmin_x \frac{1}{2\sigma^2}||Kx-y|| - \log p(x),
\end{equation}
where $p(x)$ is called the image prior and it indicates the probability of the solution $x$.
This objective can be written in the form of
\begin{align}
\argmin_x &\frac{1}{2\sigma^2}||Kx-y|| - \log p(z) \\
\text{ sbj. to } & z=x, \nonumber
\end{align}
where we substitute a new variable $z$ in place of $x$ in the prior term.
This objective has the following augmented Lagrangian~\cite{venkatakrishnan2013plug}:
\begin{equation}
\frac{1}{2\sigma^2}||Kx-y|| - \log p(z) + \frac{\rho}{2}||x-z+\lambda|| - \frac{\rho}{2}||\lambda||,
\end{equation}
where $\lambda$ represents the Lagrange multipliers.
Using the ADMM approximation, we can optimize the Lagrangian by iteratively solving the following objectives:
\begin{align}
\label{eq:minx}
\hat{x} &:= \argmin_x \frac{1}{2\sigma^2}||Kx-y|| + \frac{\rho}{2}||x-\hat{z}+\lambda||\\
\label{eq:minz}
\hat{z} &:= \argmin_z  - \log p(z)+ \frac{\rho}{2}||\hat{x}+\lambda-z||\\
\label{eq:updatelam}
\lambda &:= \lambda + (\hat{x}- \hat{z}).
\end{align}
This strategy has the benefit that the first Equation~\eqref{eq:minx} is quadratic and thus has a closed form solution:
\begin{equation}
\label{eq:minxclosedform}
\hat{x} := \big(K^TK+\sigma^2\rho \big)^{-1} \big(K^Ty+\sigma^2\rho(\hat{z}-\lambda) \big),
\end{equation}
which is often solved very efficiently in the Frequency domain as
\begin{equation}
\label{eq:minxclosedformFFT}
\hat{x} = \mathscr{F}^{-1}\left(\frac{\mathscr{F}(\bar{K}) \cdot \mathscr{F}(y) + \sigma^2\rho \mathscr{F}(\hat{z}-\lambda)}{\mathscr{F}(K^TK) + \sigma^2\rho}\right),
\end{equation}
where $\mathscr{F}$ denotes the Fourier transform, $\mathscr{F}^{-1}$ its inverse, and $\bar{K}$ a flipped and conjugate of the kernel $K$.

Another attractive point of using the ADMM approximation is that Equation~\eqref{eq:minz} represents a MAP denoising problem, where the objective is to denoise $\hat{x}+\lambda$ given the noise variance $\sigma_R^2 = \frac{1}{\rho}$.
This means that no matter the properties of the degradation model (such as blur kernel and noise variance) in Equation~\eqref{eq:minx}, we can guide its prior optimization using a separate denoising problem in Equation~\eqref{eq:minz}.
Next, we show how we can learn an end-to-end MAP denoiser to optimize the objective  and to preserve the explicitness of our natural image prior.

\section{\uppercase{Denoising Autoencoders with MMSE Objective}}
\label{sec:MMSEdenoising}
\noindent The ADMM or HQS approximations, have allowed the MAP optimization to be solved efficiently using a denoiser.
Given the fact that CNN denoisers have achieved state-of-the-art results~\cite{zhang2017beyond}, one might use these denoisers to solve Equation~\eqref{eq:minz}.
A denoising auto encoder (DAE) $R$ is trained to minimize the following MMSE loss:
\begin{equation}
\mathcal{L}_{MSE} = \sum_x ||R(x+\eta_R) - x||,
\end{equation}
where $\eta_R \sim \mathcal{N}(0,\sigma_R^2)$.
Simonceli have showed that the minimizer of this loss is the local mean of the distribution, weighted by the noise distribution. I.e.
\begin{equation}
R^*(y) = y + \sigma_R^2 \nabla \log \bar{p}(y),
\end{equation}
where $\bar{p}$ indicates the density estimate of the natural images, using a Gaussian kernel with the bandwidth $\sigma_R^2$.
Therefore, we can see that the optimal MMSE denoiser does not correspond to the solution of the MAP denoiser, i.e. Equation~\eqref{eq:minz}.
However, the difference between input and output of an optimal MMSE denoiser $R^*$, captures the gradient underlying distribution~\cite{zhang2017beyond}
\begin{equation}
R^*(y) - y = \sigma_R^2 \nabla \log \bar{p}(y),
\end{equation}
where $\bar{p}$ indicates the density estimate of the natural images, using a Gaussian kernel with the bandwidth $\sigma_R^2$.
This property has been used before~\cite{bigdeli2017deep} to optimize the MAP objective using iterative gradient descent updates.
As we show next, we can still benefit from this property to train a MAP denoiser for Equation~\eqref{eq:minz}.

\begin{figure}[t]
	\centering
	{\epsfig{file = 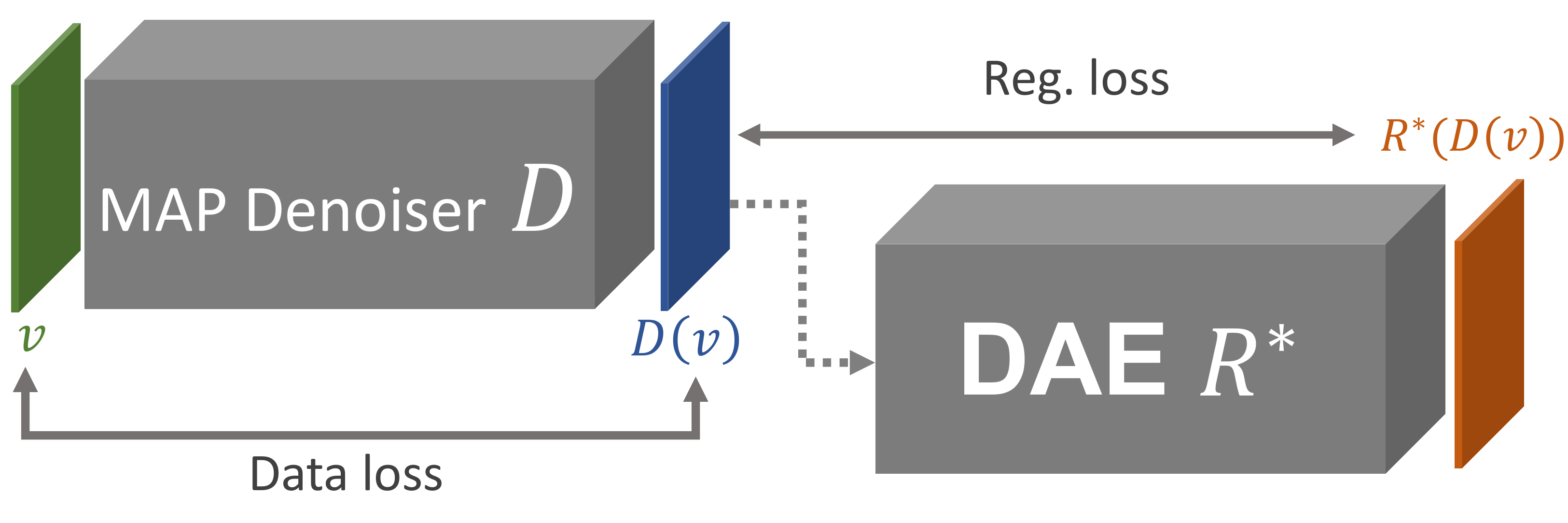, width = \linewidth}}
	\caption{Our end-to-end training of MAP denoisers using trained DAEs: The DAE R is used to enforce the regularization term for optimizing D. At the same time, the output of the network D is kept close to its input using the data loss.}
	\label{fig:networks}
\end{figure}

\section{\uppercase{Learning MAP Denoisers}}
\label{sec:denoising}

\noindent As shown before, the derived denoiser in Equation~\eqref{eq:minz} is a MAP estimator and using a MMSE denoiser would instead implicitly change the prior from the natural image distribution to an unknown representation.
To preserve the correctness of the prior, we propose to train and use a MAP denoiser instead.

We aim to train a network $D$ that minimizes the objective in Equation~\eqref{eq:minz} by
approximating the true distribution with its kernel estimate as above.
Inspired by the results of optimal DAEs, we formulate the following loss for our map denoisers,
\begin{equation}
\mathcal{L}_{MAP} = \sum_v \frac{1}{\sigma_R^2}||\bar{v}-\bar{\bar{v}}|| + \frac{\rho}{2}||\bar{v}-v||,
\label{eq:mapdae}
\end{equation}
where $\bar{v} = D(v)$ and $\bar{\bar{v}} = R^*(\bar{v})$ is the precomputed output of the optimal MMSE DAE $R^*$.

\begin{proposition}
Let us rewrite the cost function in Equation~\eqref{eq:minz} as $f(x, v)= -\log p(x)+ \alpha||x-v||$.
Denoting the gradient of the MAP-DAE loss (Equation ~\ref{eq:mapdae}) with respect to the estimator's output as $\nabla_{\bar{v}} \mathcal{L}_{MAP}$, then
\begin{align}
f(\bar{v}, v)  > f(\bar{v} -\epsilon \nabla_{\bar{v}} \mathcal{L}_{MAP}, v),
\end{align}
for small enough step size $\epsilon$.
\label{prop:mapdae}
\end{proposition}

\begin{proof}
We write the gradient of Equation~\eqref{eq:mapdae} as,
\begin{align}
\nabla_{\bar{v}} \mathcal{L}_{MAP} &=  \frac{1}{\sigma_R^2}( \bar{v}-\bar{\bar{v}} ) + \frac{\rho}{2}(\bar{v}-v), \\
&=  \frac{1}{\sigma_R^2}\big(\bar{v}-R^*(\bar{v})\big) + \frac{\rho}{2}(\bar{v}-v), \\
&=  -\nabla_{\bar{v}} \log \bar{p}(\bar{v}) + \frac{\rho}{2}(\bar{v}-v), \\
&=  \nabla_{\bar{v}} f(\bar{v}, v),
\end{align}
which means that the gradients of the loss in Equation~\eqref{eq:mapdae} are equal to the gradients of the cost function in Equation~\eqref{eq:minz}.
Therefore, by minimizing the MAP-DAE loss (Equation ~\ref{eq:mapdae}) with a small enough step size $\epsilon$, we will also minimize the cost function in Equation~\eqref{eq:minz}.
\end{proof}

This means that we can train a network to optimize Equation~\eqref{eq:minz} by minimizing $\mathcal{L}_{MAP}$ with respect to the parameters of the network $D$.
Additionally, this end-to-end training is performed without requiring any paired noisy images and their corresponding ground truth images.
Note that the minimization problem in Equation~\eqref{eq:minz} is non-convex (due to the inherent complexity of the underlying natural image distribution) and any optimization scheme can only guarantee convergence to the local minima.
However, training a neural network to optimize this objective over all images in the dataset helps in generalization~\cite{zhang2016understanding}, which leads to getting closer to the global optimum.

For the inference time, we feed our network with value $v = \hat{x}+\lambda$, which gives the denoised results $\hat{z} = D^*(v)$ based on the MAP objective Equation~\eqref{eq:minz}.
The resulting intermediate variable is then used in the rest of the ADMM optimization, which we describe next.

\section{\uppercase{Experiments}}
\label{sec:experiments}

\noindent In this section we describe out implementation details and the experimental setup for two image restoration problems: image deblurring and inpainting.
We compare our method with the state-of-the-art using the structural similarity image measure (SSIM) and peak signal to noise ratio (PSNR) measures.
\subsection{Network Architecture and Training}

We have parametrized our denoiser networks based on the DnCNN architecture~\cite{zhang2017beyond}.
We employ this network for both the MAP and the DAE denoisers.
The networks consist of 17 convolutional layers with kernel of size $3\times3$, 64 feature channels with ReLU non-linearities.
The final receptive field of our network is then $35\times35$.
Unlike the original DnCNN training, we did not see any benefit of using batch normalization in our networks, and omit this in our final models.

The training of the two networks are independent: to train the MAP denoiser network, we need an optimally trained DAE.
The DAE training is also independent, since it requires noisy input samples generated from our clean dataset.
In practice, we train the two networks at the same time, taking advantage of the parallelization to speed up the training.
We make sure that the DAE is converged long enough before stopping the MAP denoiser training.
In defining the loss terms of the MAP denoiser, we also make sure that the gradients of the regularization term are not propagated through the DAE.

Similar to prior work~\cite{bigdeli2017image}, we observed that the DnCNN network only performs well if it is used with noisy input images, with the same noise standard deviation $\sigma_R$ as in training.
Consequently, we also add noise with standard deviation $\sigma_R$ to network output $D(v)$ before passing it to the DAE $R$.
This has been shown to approximate the actual objective by minimizing an upper-bound of the actual objective.

We have trained the proposed MAP denoiser on the same dataset as the original DnCNN~\cite{zhang2017beyond}.
We used training patches of size $40\times40$ which were cropped from the Berkeley segmentation dataset~\cite{berkeley} originally consisting of 400 gray-scale images.

\subsection{MAP Denoiser Evaluation}
We have evaluated the behaviour of the MAP denoiser using visual inspections.
Figure~\ref{fig:denoiser} shows the results of the network when applied to a noise-free image.
The network moves the input towards the more likely regions of the natural image distribution be removing unlikely patterns.
This is done by progressively cancelling out some high frequency details, such as noise or small edges, while preserving the sharpness of the more stronger edges.
This leads to the visually cleaner results, which can be controlled by the number of iterations.
We can achieve visually cleaner results by controlling  the number of update iterations.

\begin{figure}[t]
	\centering
	{\epsfig{file = 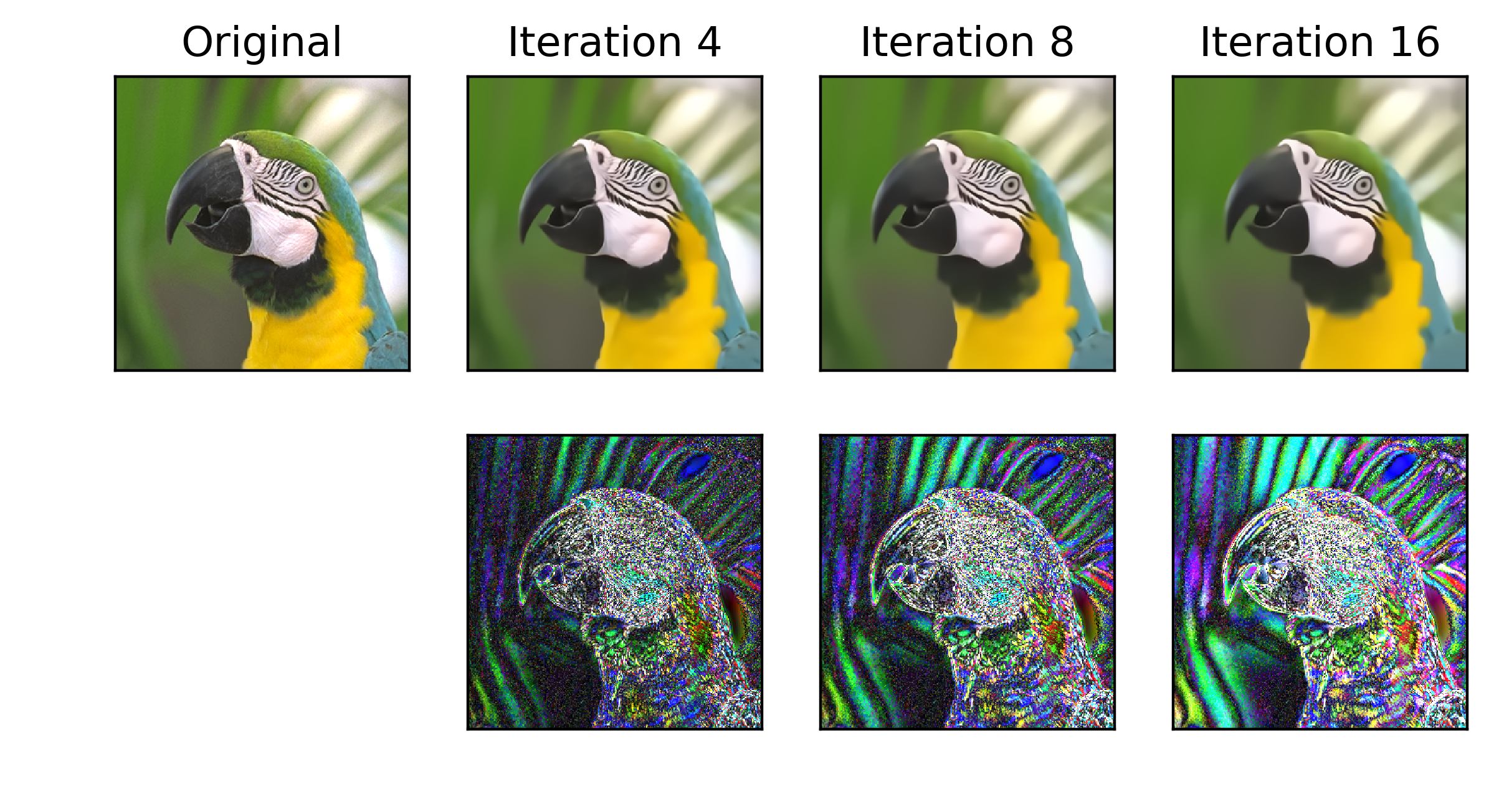, width = \linewidth}}
	\caption{Iterative results of our MAP denoiser network (top row) and the residual w.r.t. the input image (bottom row). }
	\label{fig:denoiser}
\end{figure}

\subsection{Non-Blind Image Deblurring}
In this work we focus on application of the proposed image restoration method to non-blind image deblurring. We follow the ADMM approach for image restoration presented in Section~\ref{sec:restoration}.
We summarize our approach in Algorithm~\ref{algo:admm}, where each iteration of the algorithm is as follows:
In the first step we optimize $\hat{x}$.
For image deblurring this step can be done efficiently in frequency domain using Equation~\eqref{eq:minxclosedformFFT}.
Second, we compute $\hat{z}$ by a single feed-forward step using our trained MAP denoiser $D^*$.
And finally, we update the variable $\lambda$ and reiterate until convergence.

\begin{algorithm}[b]
\SetKwInOut{Input}{input}\SetKwInOut{Output}{output}
 \caption{Optimization steps for non-blind image deblurring.}
 \Input{Degraded image $y$, blur kernel matrix $K$, and noise standard deviation $\sigma$}
 \While{not converged}{
 \textbf{1.} \small{$\hat{x} = \big(K^TK+\sigma^2\rho \big)^{-1} \big(K^Ty+\sigma^2\rho(\hat{z}-\lambda) \big)$} \\
 \textbf{2.} $\hat{z} = D^*(\hat{x} + \lambda)$ \\
 \textbf{3.} $\lambda = \lambda + (\hat{x}- \hat{z})$
 }
 \Output{MAP estimate output image $\hat{x}$}
 \label{algo:admm}
\end{algorithm}

We have conducted several experiments on various datasets to find the best hyper-parameters for grayscale non-blind image deblurring.
We found that using the DAE standard deviation $\sigma_R = 7$ performs best in practice.
We also found that setting $\rho = 1/\sigma_R^2$, so that the Equation~\eqref{eq:minx} is balanced, performs best in practice.
Such setting leads to best results on Sun dataset~\cite{sundataset}, however, the optimal value for each image might be different.
In the experiments we use 75 iterations of the ADMM algorithm. However, as shown in Figure~\ref{fig:diff}, 35 iterations is practically enough for most experiences.

The processing speed and the convergence rate is very important in all iterative approaches.
As shown in Figure~\ref{fig:diff}, the proposed method converges very fast compared to DMSP (both methods are solving the same MAP optimization).
A color image from BSDS300 can be processed using 35 iterations in about 0.8 seconds on an NVIDIA GTX2080 GPU including the data transfer between the host computer and the device.
This makes our method very attractive and more useful in practice.
We also show the convergence speed of IRCNN, where the noise level of the denoiser is exponentially decayed from 49 to 15. We note, that the convergence speed of IRCNN may be significantly increased at the cost of lower PSNR.

\begin{figure}[t]
	\centering
	{\epsfig{file = 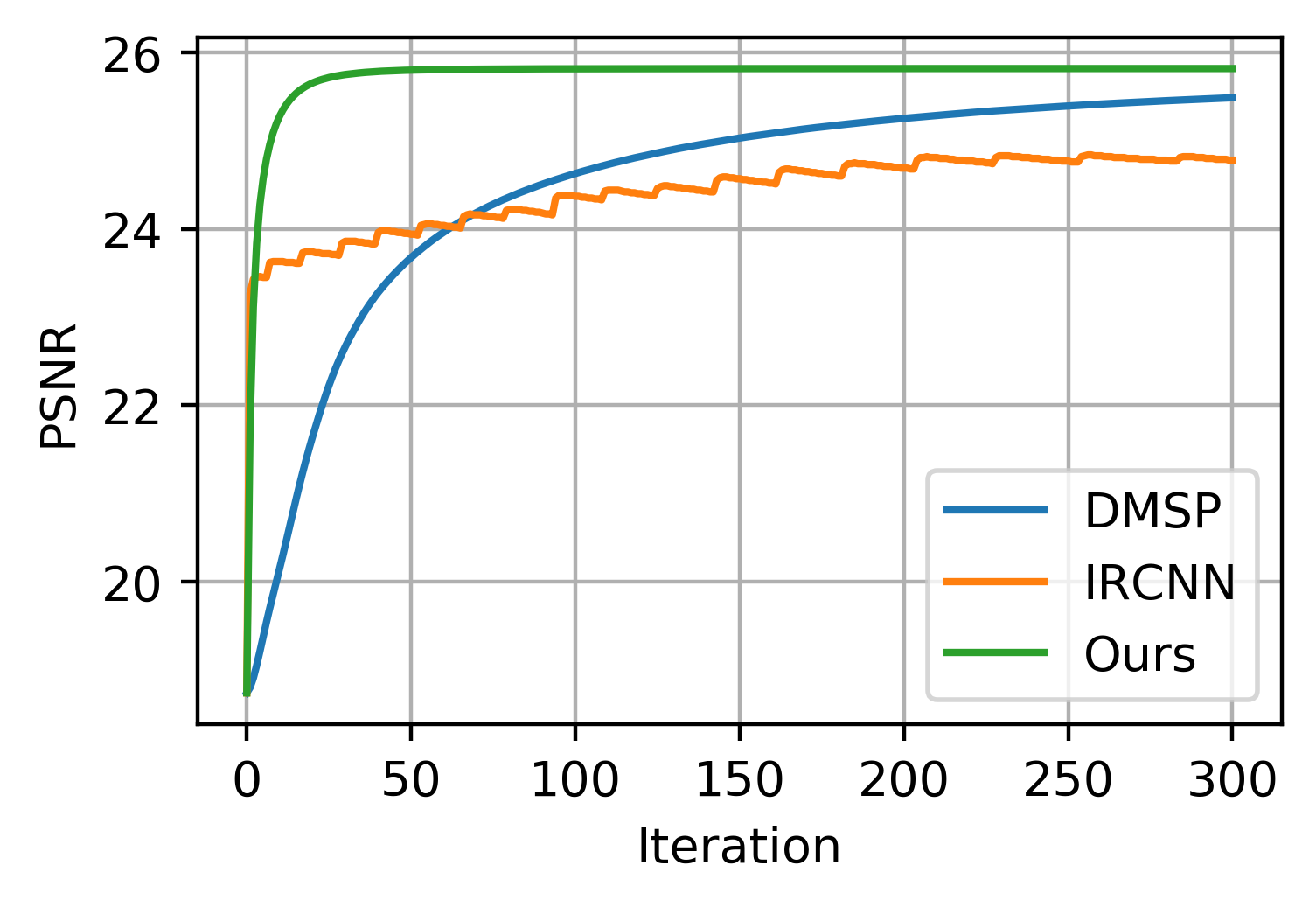, width = \linewidth}}
	\caption{Convergence speed comparison on an image from BSDS300. Note that DMSP and the proposed method try to optimize the same objective. Using the ADMM approach, we can speed-up the optimization up to 70x faster.}
	\label{fig:diff}
\end{figure}

\subsubsection{Optimization}
To be consistent with other methods, for all the experiments we first blur the ground-truth test image with the blur kernel. The blurring is done by convolution with a flipped version of a kernel and only the valid area is preserved. We also add a Gaussian noise with the desired standard deviation $\sigma$.
After the restoration, we measure the PSNR on the valid area.

Since we optimize $\hat{x}$ in the frequency domain, the proposed method works the best when the degradation is done using circular convolution. However, in case we receive only the valid area, assuming the circular convolution leads to high error on the boarders which gets propagated towards the center in subsequent iterations.
To alleviate this behaviour and maintain the benefit of processing in frequency domain, we first pad the degraded image by replicating the edges to the original size. Furthermore, in each iteration, we first estimate the non-valid area of the $y$ by performing a cyclic convolution on the current estimate of $\hat{x}$ with the kernel.
Such treatment leads to plausible results and yet it is fast enough due to the processing in the frequency domain.
Figure~\ref{fig:deblurring} shows the iterative results of our optimization, where the image is getting sharper using the proposed denoising prior.

\begin{figure*}[t]
	\centering
	{\epsfig{file = 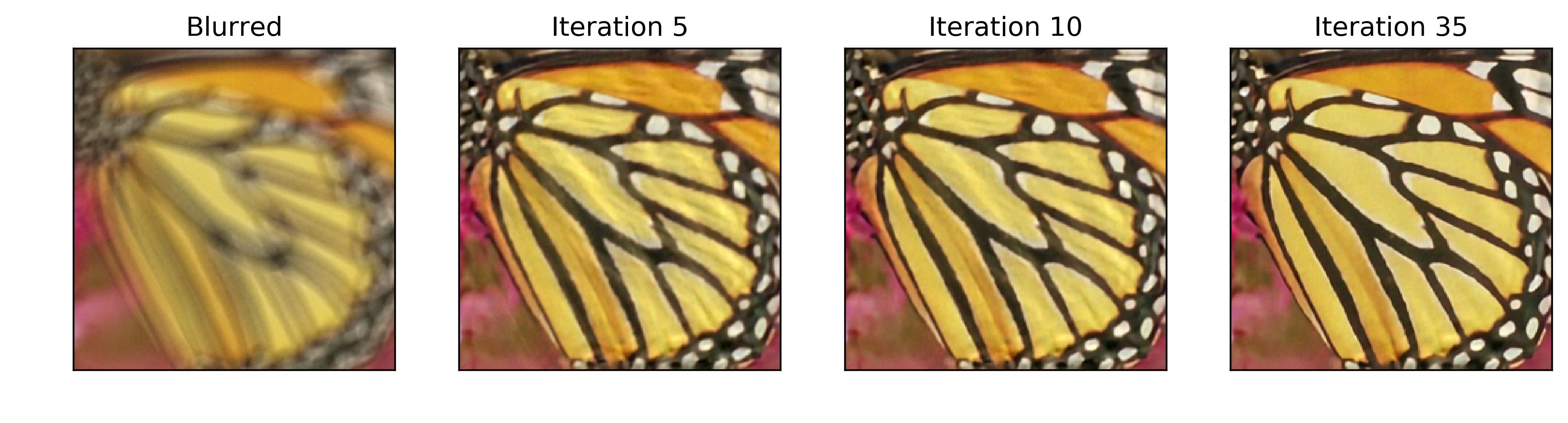, width = \linewidth}}
	\caption{Deblurring results of the ADMMM iterations using the proposed algorithm. Our MAP denoiser network encourages sharp image edges and removes undesired artifacts.}
	\label{fig:deblurring}
\end{figure*}

\begin{table*}
	\caption{Quantitative comparisons for non-blind image deblurring using two datasets in terms of PSNR.}\label{tab:psnr_deblurSun} \centering
	\begin{tabular}{ l r c c c c c c c c }
		\hlineB{3}
		& \multicolumn{4}{c}{Sun~\cite{sundataset}} && \multicolumn{4}{c}{BSDS300~\cite{berkeley}} \\
		\cline{2-5} \cline{7-10}

		Method \hspace{90pt} $\sigma \rightarrow$ & 2.55 & 5.10 & 7.65 & 10.2  && 2.55 & 5.10 & 7.65 & 10.2 \\
		\hline

		FD \cite{krishnan2009fast} & 30.79 & 28.90 & 27.86 &  27.14 && 24.44 & 23.24 & 22.64 & 22.07 \\

		EPLL \cite{zoran2011learning} & \textbf{32.05} & 29.60 & 28.25 & 27.34 && 25.38 & 23.53 & 22.54 & 21.91 \\

		CSF \cite{schmidt2014shrinkage} & 30.88 & 28.60 & 27.65 & 26.97 && 24.73 & 23.60 & 22.88 & 22.44 \\

		TNRD \cite{chen2016trainable} & 30.03 & 28.79 & 28.04 & 27.54 && 24.17 & 23.76 & 23.27 & 22.87 \\

		DAEP \cite{bigdeli2017image} & 31.76 & 29.31 & 28.01 & 27.16 && 25.42& 23.67 & 22.78 & 22.21 \\

		IRCNN \cite{zhang2017learning} & 31.80 & \textbf{30.13} & 28.93 & 28.09 && 25.60& 24.24 & 23.42 & 22.91 \\

		GradNet 7S \cite{jin2017noise} & 31.75 & 29.31 & 28.04 & 27.54 && 25.57 & 24.23 & 23.46 & 22.94 \\

		DMSP \cite{bigdeli2017deep} & 29.41 & 29.04 & 28.56 & 27.97 && 25.69 & 24.45 & \textbf{23.60} & \textbf{22.99} \\

		Ours & 31.00 & 29.96 & 2\textbf{8.96} & \textbf{28.13} && \textbf{26.18} & \textbf{24.52} & 23.51 & 22.79 \\
		\hlineB{3}

	\end{tabular}
\end{table*}

\begin{table*}
	\vspace{0.2cm}
	\caption{Comparison of SSIM scores for non-blind image deblurring using 	two datasets.}\label{tab:ssim_deblurSun} \centering
	\begin{tabular}{ l r c c c c c c c c }
		\hlineB{3}
		& \multicolumn{4}{c}{Sun \cite{sundataset}} && \multicolumn{4}{c}{BSDS300 \cite{berkeley}} \\
		\cline{2-5} \cline{7-10}
		
		Method \hspace{90pt} $\sigma \rightarrow$ & 2.55 & 5.10 & 7.65 & 10.2  && 2.55 & 5.10 & 7.65 & 10.2 \\
		\hline
		
		FD \cite{krishnan2009fast} & 0.851 & 0.787 & 0.744 &  0.714 && 0.664 & 0.577 & 0.534 & 0.492 \\
		
		EPLL \cite{zoran2011learning} & \textbf{0.880} & 0.807 & 0.758 & 0.721 && 0.712 & 0.590 & 0.521 & 0.476 \\
		
		CSF \cite{schmidt2014shrinkage} & 0.853 & 0.752 & 0.718 & 0.681 && 0.693 & 0.612 & 0.558 & 0.521 \\
		
		TNRD \cite{chen2016trainable} & 0.844 & 0.790 & 0.750 & 0.739 && 0.690 & 0.631 & 0.589 & 0.550 \\
		
		GradNet 7S \cite{jin2017noise} & 0.873 & 0.798 & 0.750 & 0.733 && 0.731 & 0.653 & 0.595 & 0.552 \\
		
		DMSP \cite{bigdeli2017deep} & - & - & - & - && 0.740 & \textbf{0.671} & \textbf{0.611} & 0.563 \\
		
		Ours & 0.829 & \textbf{0.817} & \textbf{0.788} & \textbf{0.758} && 0.768 & 0.668 & 0.595 & 0.540 \\
		\hlineB{3}
		
	\end{tabular}
\end{table*}

\subsubsection{Datasets}
We test the method on two different datasets and report the results in Table~\ref{tab:psnr_deblurSun} and Table~\ref{tab:ssim_deblurSun}.
The first one is the Sun dataset~\cite{sundataset} that consists of 80 images with the longest dimension being $924$ pixels.
The dataset contains 8 blur kernels of various sizes to simulate the degraded images.
The second dataset we use for evaluation is the BSDS300~\cite{berkeley}.
It consists of 300 images of size $321\times481$, which we convert to gray-scale.
We test this dataset with the 5 large blur kernels as in~\cite{jinkernels}.
These kernels visualized in Figure~\ref{fig:kernels}, have complex structures and lead to challenging deblurring tasks.

\begin{figure}[b]
	\centering
	{\epsfig{file = 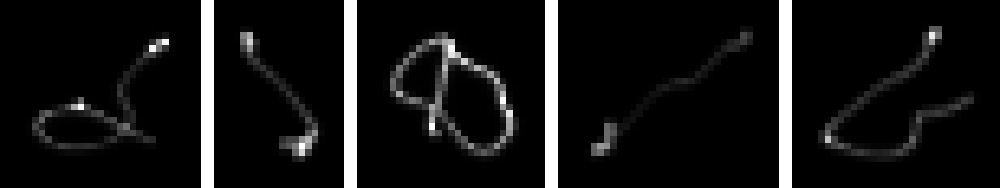, width = \linewidth}}
	\caption{Blur kernels used in our experiments \cite{jinkernels}.}
	\label{fig:kernels}
	\vspace{0.1cm}
\end{figure}

\begin{figure*}[t!]
	\centering
	\begin{subfigure}[t]{0.15\textwidth}
		\centering
		{\epsfig{file = 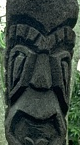, width = \linewidth}}
		\caption{Ground truth}
	\end{subfigure}%
	~ 
	\begin{subfigure}[t]{0.15\textwidth}
	\centering
	{\epsfig{file = 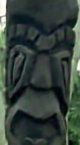, width = \linewidth}}
	\caption{EPLL}
    \end{subfigure}%
    ~ 
	\begin{subfigure}[t]{0.15\textwidth}
	\centering
	{\epsfig{file = 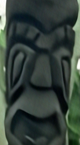, width = \linewidth}}
	\caption{DAEP}
	\end{subfigure}%
	~ 
	\begin{subfigure}[t]{0.15\textwidth}
	\centering
	{\epsfig{file = 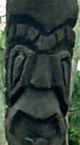, width = \linewidth}}
	\caption{GradNet~7S}
	\end{subfigure}%
	~ 
	\begin{subfigure}[t]{0.15\textwidth}
	\centering
	{\epsfig{file = 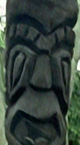, width = \linewidth}}
	\caption{DMSP}
	\end{subfigure}%
	~ 
	\begin{subfigure}[t]{0.15\textwidth}
		\centering
		{\epsfig{file = 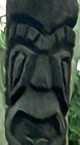, width = \linewidth}}
		\caption{Ours}
	\end{subfigure}
	\caption{Visual Comparison of the selected methods for the task of non-blind image deblurring.}
	\label{fig:debluringComparison}
\end{figure*}

Figure~\ref{fig:debluringComparison}, shows a visual comparison between our method and several other state-of-the-art techniques.
Our method is able to capture the sharp structures of the image, while remaining highly efficient compared to other algorithms.

\subsection{Image Inpainting}
In addition to the deblurring method, we have also tested the proposed method on the inpainting problem.
We follow the ADMM approach for image restoration presented in Section~\ref{sec:restoration}.
To speed up the convergence, we initialize our solution $\hat{x}$ with image inpainted using a median filter.

\subsubsection{Optimization}
We conduct the test with 80\% of the missing pixels where we add Gaussian noise with $\sigma=12$ prior to the remaining pixels.
The tested methods are provided both with the degraded image and with the mask specifying which pixels were dropped.

In contrast to the image deblurring, we use 300 iterations of the ADMM algorithm. Moreover, in each iteration, we use 200 steps of gradient descent optimization to approximate Equation~\eqref{eq:minx} of the ADMM.

\subsubsection{Dataset}
We evaluate the methods using PSNR and SSIM scores on classical dataset consisting of images of \textit{cameraman}, \textit{house}, \textit{peppers}, \textit{Lena}, \textit{Barbara}, \textit{boat}, \textit{hill}, and \textit{couple}.

Tables~\ref{tab:psnr_inpaint12},~\ref{tab:ssim_inpaint12} show the PSNR and SSIM scores of our method compared to the state-of-the-art.
Again, our method can achieve state-of-the-art results for the task of image inpainting.
Note that we use the same MAP denoising network as in the previous experiments.
We also visualize qualitative results in Figure~\ref{fig:inpainting}, where our method can successfully recover 80\% of the missing pixels from the input image.

\begin{table*}[t]
	\footnotesize
	\vspace{0.2cm}
	\caption{Inpainting results (PSNR) for 80\% of missing pixels and noise with std. deviation of $\sigma=12$.}\label{tab:psnr_inpaint12} \centering
	\begin{tabular}{ l c c c c c c c c }
		\hlineB{3}
		Method \hspace{20pt} & cameraman & house & peppers & Lena  & Barbara & boat & hill  & couple \\
		\hline

		P\&P-BM3D\scriptsize{~\cite{venkatakrishnan2013plug} }& 24.43 & 30.78 & 26.56 & 29.47 & 24.12 & 26.53 & 27.44 & 26.71\\
		IRCNN~\cite{zhang2017learning} & \textbf{24.59} & 30.19 & 26.94 & 29.52 & \textbf{25.49} & 26.58 & 27.55 & 26.62\\
		IDBP-BM3D~\cite{tirer2018image} & 24.51 & \textbf{31.14} & 26.79 & 29.69 & 25.06 & 26.64& 27.61 & 26.77\\
		IDBP-CNN~\cite{tirer2018image}  & 24.14 & 30.92 & \textbf{27.17} & \textbf{29.80} & 23.61 & 26.78 & 27.70 & \textbf{26.80}\\
		Ours & 23.17 & 30.32 & 26.01 & 29.44 & 24.30 & \textbf{26.81} & \textbf{28.40} &  26.77 \\
		\hlineB{3}
	\end{tabular}
\end{table*}

\begin{table*}[t]
	\footnotesize
	\caption{Inpainting results (SSIM) for 80\% of missing pixels and noise with std. deviation of $\sigma=12$.}\label{tab:ssim_inpaint12} \centering
	\begin{tabular}{ l c c c c c c c c }
		\hlineB{3}
		Method \hspace{20pt} & cameraman & house & peppers & Lena  & Barbara & boat & hill  & couple \\
		\hline
		
		P\&P-BM3D\scriptsize{~\cite{venkatakrishnan2013plug} } & 0.774 & 0.839 & 0.807 & 0.818 & 0.705 & 0.707 & 0.683 & 0.734\\
		IRCNN~\cite{zhang2017learning} & 0.781 & 0.835 & 0.813 & 0.820 & \textbf{0.758} & 0.723 & 0.706 & 0.736\\
		IDBP-BM3D~\cite{tirer2018image} & 0.775 & \textbf{0.844} & 0.816 & 0.824 & 0.738 & 0.712 & 0.691 & 0.738 \\
		IDBP-CNN~\cite{tirer2018image}  & \textbf{0.786 }& 0.843 & \textbf{0.830} & \textbf{0.836} & 0.731 & \textbf{0.738} & \textbf{0.714} & \textbf{0.752}\\
		Ours & 0.757 & 0.808 & 0.814 & 0.803 & 0.741 & 0.721 & 0.707 &  0.735 \\
		\hlineB{3}
	\end{tabular}
\end{table*}

\section{\uppercase{Conclusions}}
\label{sec:conclusion}

\noindent We discussed the challenges of using the plug-and-play denoising in ADMM optimization for MAP estimation.
We presented an approach for learning an end-to-end neural network to perform MAP denoising.
We used this network in an ADMM framework to perform generic image restoration tasks.
Our theoretical results show that we can guarantee to minimize the MAP objective using the proposed training strategy.
And our experimental validation showed that our method has significant improvement in speeding up the MAP optimization, compared to other explicit approaches.
Over all, other method supports the theoretical guarantees of MAP estimation, and at the same time benefits from the fast performance of other approaches without such guarantee.
\vfill
\section*{\uppercase{Acknowledgements}}
\noindent Financial support from the CSEM Data Program is grateully acknowledged.
Authors would also like to thank Laura Bujouves for her insightful comments. 

\begin{figure*}[h!]
	\vspace{0.1cm}
	\centering
	{\epsfig{file = 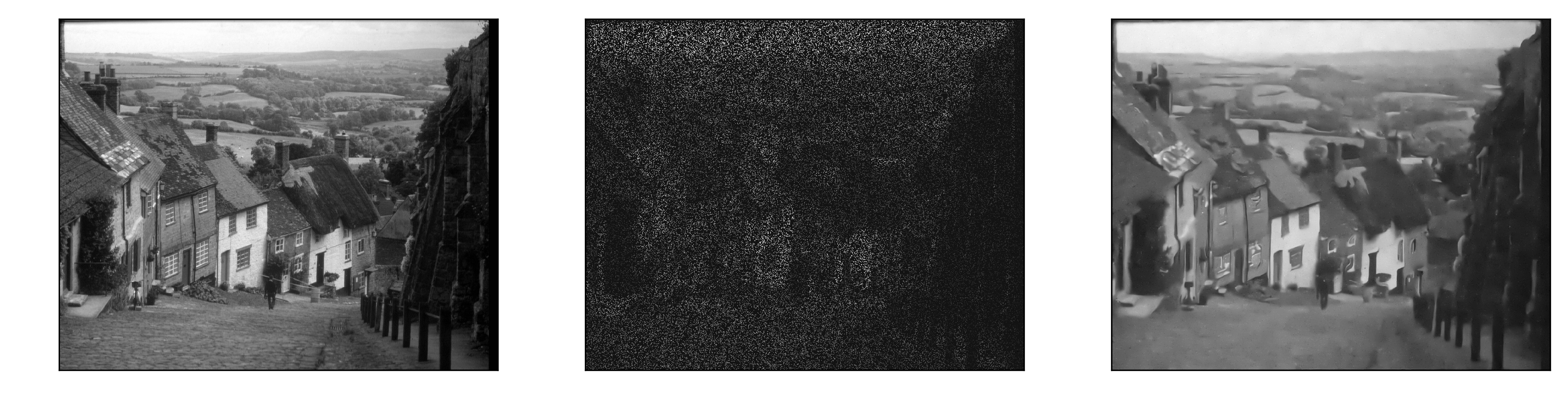, width = \linewidth}}
	\caption{Image inpainting example. \textbf{Left:} Original image, \textbf{Middle:} Degraded image with 80\%  of the pixels set to zero, \textbf{Right:} Inpainted image using our method.}
	\label{fig:inpainting}
\end{figure*}

\bibliographystyle{apalike}
{\small
\bibliography{camera_ready}}

\end{document}